\documentclass[twoside]{IEEEtran}
\usepackage{bbm,dsfont,mathrsfs,fixmath,amsthm}
\usepackage{amsmath,amssymb,eucal,graphicx}
\usepackage{stmaryrd,cite}
\usepackage{amsmath,amstext,amsfonts,amssymb}
\usepackage{epsfig}
\usepackage{exscale}
\usepackage{enumerate}
\usepackage{mathtools}
\usepackage{multirow}

\newtheorem{proposition}{Proposition}%[section]
\newtheorem{theorem}{Theorem}
\newcommand{\ep}{\hfill $\Box$}

\newcommand{\cond}{\,\vert\,}

\newcommand{\defeq}{\triangleq}

\newcommand{\indic}{\mathbf{1}}

\newfont{\bbb}{msbm10 scaled 500}

\newfont{\bb}{msbm10 scaled 1100}
\newcommand{\CC}{\mbox{\bb C}}
\newcommand{\RR}{\mbox{\bb R}}

\newcommand{\EE}{\mbox{\bb E}}

\newcommand{\ev}{{\bf e}}

\newcommand{\hv}{{\bf h}}

\newcommand{\rv}{{\bf r}}

\newcommand{\Rm}{{\bf R}}

\newcommand{\Jc}{{\cal J}}
\newcommand{\Kc}{{\cal K}}

\newcommand{\Nc}{{\cal N}}

\newcommand{\alphav}{\hbox{\boldmath$\alpha$}}

\newcommand{\etav}{\hbox{\boldmath$\eta$}}

\newcommand{\piv}{\hbox{\boldmath$\pi$}}

\renewcommand{\arg}{{\hbox{arg}}}

\newcommand{\eqdef}{\stackrel{\Delta}{=}}

\DeclareFontFamily{U}{cmfi}{}
\DeclareFontShape{U}{cmfi}{m}{n}{ <-> cmfi10 }{}
\DeclareSymbolFont{CMFI}{U}{cmfi}{m}{n}

\renewcommand{\Rm}{\pmb{R}}

\renewcommand{\ev}{\pmb{e}}

\renewcommand{\hv}{\pmb{h}}

\renewcommand{\rv}{\pmb{r}}

\newcommand{\Rbl}{R_{\rm sum,bl}}	% baseline
\newcommand{\Rsp}{R_{\rm sum,sp}}	% superposition
\newcommand{\Rsc}{R_{\rm sum,sc}}	% scheduling (user selection)

\newif\ifarxiv
\arxivtrue

\begin{document}

\title{Opportunistic Content Delivery in Fading Broadcast Channels}

\author{\IEEEauthorblockN{}
\IEEEauthorblockA{A. Ghorbel, K. H. Ngo, R. Combes, M. Kobayashi, and S. Yang \\
LSS, CentraleSup\'elec \\
Gif-sur-Yvette, France\\
 {\tt \{firstname.lastname\}@centralesupelec.fr}
}
\thanks{This work was supported by Huawei Technologies France SASU.}
}
%%%%%
%\author{
%A. Ghorbel, K. H. Ngo, R. Combes, M. Kobayashi, and S. Yang \\
%LSS, CentraleSup\'elec \\
%Gif-sur-Yvette, France\\
% {\tt \{firstname.lastname\}@centralesupelec.fr} 
%\thanks{This work was supported by Huawei Technologies France SASU. %partially 
%}
%}
%%%

\maketitle
\begin{abstract}
We consider content delivery over fading broadcast channels. A server wants to transmit $K$ files to $K$ users, each equipped with a cache of finite size. Using the coded caching scheme of Maddah-Ali and Niesen, we design an opportunistic delivery scheme where the long-term sum content delivery rate scales with $K$ the number of users in the system. 
The proposed delivery scheme combines superposition coding together with appropriate power allocation across sub-files intended to different subsets of users. We analyze the long-term average sum content delivery rate achieved by two special cases of our scheme: a) a selection scheme that chooses the subset of users with the largest weighted rate, and b) a baseline scheme that transmits to $K$ users using the scheme of Maddah-Ali and Niesen. We prove that coded caching with appropriate user selection is scalable since it yields a linear increase of the average sum content delivery rate. 
\end{abstract}

\section{Introduction}\label{section:introduction}
Content delivery applications such as video streaming are envisionned to represent nearly 75\% of the mobile data traffic by 2020. 
The skewness of the video traffic together with the ever-growing cheap on-board storage memory suggests
that the quality of experience can be improved by caching popular content close to the end-users in wireless networks. 
Recent works have studied the gains provided by caching under various models and assumptions (see e.g. \cite{maddah2013fundamental,golrezaei2011femtocaching} and references therein). 
%Most of existing works assume that caching is performed in two phases: {\it placement phase} to prefetch users' caches under their memory constraints (typically during off-peak hours) prior to the actual demands; {\it delivery phase} to transmit codewords 
%such that each user, based on the received signal and the contents of its cache, is able to decode its requested file. 
In this work, we consider content delivery using coded caching where a server is connected to $K$ users each equipped with a cache of finite memory \cite{maddah2013fundamental}. A striking result of \cite{maddah2013fundamental} is that the total number of multicast transmissions to satisfy $K$ distinct requests converges to a constant in the regime of a large $K$, thus yielding a scalable system. 

Substantial effort have been devoted to quantify the gains of coded caching in more realistic scenarios (see e.g. \cite[Section VIII]{maddah2013fundamental}, \cite{misconceptions}). In particular, some authors have studied coded caching over wireless channels by relaxing the initial 
assumption of a perfect shared link between the server and users  \cite{bidokhti2016noisy,NgoAllerton2016,zhang2016wireless,Shariatpanahi2016multiserver,Shariatpanahi2016multiantenna}.
It is noted that the performance of coded caching strongly depends on the multicast rate of the underlying wireless channels and the latter is limited by the user in the worst channel condition. Such limitation has been highlighted in~\cite{NgoAllerton2016} which shows that the sum content delivery rate is no longer scalable, if the multicast rate vanishes when $K \to \infty$. This is typically the case for the i.i.d. Rayleigh fading channels (i.i.d. across users and time) \cite{jindal2006capacity}.

In a more realistic scenario where users have asymmetric fading statistics (e.g. in a cellular system), the performance degradation becomes substantial in the sense that most of the resources are allocated to users with low channel quality. To overcome these drawbacks, schemes using multiple antennas \cite{Shariatpanahi2016multiserver,NgoAllerton2016,Shariatpanahi2016multiantenna} and interference management techniques \cite{bidokhti2016noisy,zhang2016wireless} have been proposed. In this work, we take a different approach based on user scheduling in order to address the following fundamental question that has been overlooked in existing works: {\it how to exploit the wireless channels opportunistically for content delivery?} 

To answer this question, we consider the $K$-user Gaussian fading broadcast channel with $2^K-1$ independent messages, each intended to a subset of users, and solve the weighted sum rate maximization in section~\ref{section:formulation}. The optimal strategy combines superposition coding with an appropriate power allocation across different messages. The solution at hand can be applied to various communication contexts such as a queued content delivery network \cite{SubmittedWiopt2016}. We apply this solution to maximize the sum content delivery rate, assuming that content placement is performed by existing schemes \cite{maddah2013fundamental,maddah2013decentralized}. 
We analyze the performance of our scheme in two special cases of interest: a) a selection scheme that chooses the subset of users with the largest instantaneous weighted rate, and b) a baseline scheme that applies coded caching to $K$ users. %in section \ref{sec:perfom_analysis}
We prove that the selection scheme achieves a linear increase of the average sum content delivery rate in the regime of a large $K$ thus yields a scalable solution. On the other hand, both the baseline and the selection schemes achieve the same sum delivery rate in the high SNR regime, since it is nearly optimal to perform coded caching over all $K$ users in this regime.  Moreover, we provide a simple threshold-based feedback scheme which yields the same performance as the selection scheme in the large $K$ regime, while requiring each user to feedback only one bit rather than its channel state information. Numerical examples in Section \ref{sec:numericalEx} show that the linear gain in sum content delivery rate occurs even for relatively small number of users. 
 \ifarxiv
	Proofs of Theorem 1, Propositions 1 and 2 are presented in Appendix.
\else
	Due to lack of space, we provide only sketches of proof, and complete proofs can be found in the technical report~\cite{techreport}.  
\fi

We use the following notation: $[k] =\{1, \dots, k\}$, and $f(x)\sim g(x)$ if $\lim_{x\rightarrow\infty}\frac{f(x)}{g(x)}=1$.

\section{System Model}\label{section:model}

We consider a content delivery system where a server with $N$ files wishes to transmit $K$ requested files to $K$ users over a wireless downlink channel. We assume that $N$ files are of equal size of $F$ bits and equal popularity, while each user has a cache of size $M F$ bits, where $M \ge 1$ denotes the cache size measured in files. We define $m$ the normalized cache size denoted by $m=M/N$.
% We focus on the regime of a large number of files, i.e. $N\geq K$ such that the requests from $K$ users are all {\it distinct}. Further, 
Each user can store any part of any file in her cache, by prefetching them during off-peak hours, prior to the actual request, according to centralized or decentralized placement strategies proposed in the literature. 

In the decentralized placement of \cite{maddah2013decentralized}, each user independently caches a subset of $mF$ bits of file $i$, chosen uniformly at random for $i=1,\dots, N$ under the memory constraint of $MF$ bits. By letting $W_{i|\Jc}$ denote the sub-file of $W_i$ stored exclusively in the cache memories of the user set $\Jc$, the cache memory $Z_k$ of user $k$ after decentralized placement is given by
\begin{align} \label{eq:Zk}
Z_k =\{ W_{i \cond \Jc}:  \;\; \forall \Jc \subseteq[K], \forall \Jc \ni k , \forall i \in [N] \}.
\end{align}
%The size of each sub-file measured in bits is given by  
%\begin{align} 
%|W_{i \cond \Jc}|= F m^{|\Jc|}\left(1-m\right)^{K-|\Jc|}+o(F), \label{eq:LLN}
%\end{align}
%as $F\rightarrow \infty$. 
In the centralized cache placement \cite{maddah2013fundamental}, each file is split into ${K \choose b}$ disjoint sub-files of
equal size, where $b\defeq\lfloor m K\rfloor$. Each sub-file is cached by users in a subset $\Jc$ of cardinality $|\Jc|=b$. 
%The size of each sub-file is  
%\begin{align}\label{eq:cent}
% |W_{i \cond \Jc}|=\frac{F}{{K \choose b}}. 
% \end{align}  
The resulting cache memory $Z_k$ is the same as \eqref{eq:Zk} except that the subsets are now restricted to those with a specific cardinality $b$. Once the requests from users are revealed, the server generates and sequentially conveys the codewords intended to each subset of users. Namely, assuming that user $k$ requests file $k$ for all $k$, the codeword intended to the subset $\Jc$ is given by
\begin{align}\label{eq:V}
V_{\Jc}=\oplus_{k\in \Jc}W_{k|\Jc\setminus\{k\}},
\end{align}
where $\oplus$ denotes the bit-wise XOR operation. The main idea here is to create a codeword useful to a subset of users by exploiting the receiver side information established during the placement phase. 

It has been shown in \cite{maddah2013fundamental,maddah2013decentralized} that the total number of 
multicast transmissions needed to satisfy $K$ {\it distinct} demands over the error-free shared link is as follows.  
\begin{align} \label{eq:defT}
T(m,K)= \begin{cases}  
\left(1-m\right)
\frac{1}{1/K+m}, &\text{centralized caching}, \\ 
\left(1-m\right) 
\frac{1-(1-m)^K}{m}, &\text{decentralized caching}.
\end{cases} 
\end{align}%  %; $T$ is normalized by $F$, the number of bits to transmit is $T(m,K) F$. 
%Note that $T(m,K) \xrightarrow{K \to \infty} \frac{1-m}{m}$ for both centralized and decentralized coded caching. 

In the physical layer, we consider the quasi-static Rayleigh fading broadcast channel. The output of user $k$ at channel use $t$ is given by
\begin{align}\label{eq:FadingBC}
y_k[t] = \sqrt{h_k} x[t] + w_k[t],
\end{align}
where $x$ is the input symbol satisfying the power constraint $\frac{1}{n}\sum_{t=1}^{n}|x[t]|^2\leq P$; $\{h_k\}$ are the fading gains, independently and exponentially distributed $\sim $ Exp$(1/\gamma_k)$ with mean $\gamma_k$; $w_k(t)\sim \Nc_{\CC} (0, 1)$ is additive white Gaussian noise assumed independent between users. We assume that $\{h_k\}$ are known by the server and all users. 

It is well-known that the multicast capacity of the channel at hand, or the common message rate,  is given by 
\begin{align} \label{eq:multicast_rate}
R_{\rm mc} (\hv) = \log\left(1+ P \min_{j\in [K]}  h_j\right)
\end{align}
and is limited by the user in the worst fading condition. It has been proved in \cite{NgoAllerton2016} that such limitation is detrimental for a scalable content delivey network. 

To see this, let us first define the sum content delivery rate when coded caching is applied directly to the fading broadcast channel. 
In order to satisfy the distinct demands from $K$ users, that is to \emph{complete} the transfer of $K F$ demanded bits, 
one needs to send $T(m,K) F$ bits over the wireless link. The corresponding transmission takes $\frac{T(m,K) F}{ R_{\rm mc}(\hv)} $ units of time. As a result, the sum content delivery rate of a naive application of coded caching 
for a given channel realization $\hv$ is given by 
\begin{align*} %\label{eq:Rmul}
%\Rsum(K,P) 
 \frac{K}{T(m,K)} R_{\rm mc}(\hv)% \quad \text{bits/second/Hz}
\end{align*}% 
measured in [nats/second/Hz]. We call this scheme the ``baseline'' scheme, and its long-term average sum content delivery rate is
\begin{align} \label{eq:baseline_rate}
\Rbl (K)=\frac{K}{T(m,K)} \EE[R_{\rm mc}(\hv)].
\end{align}
In the case of symmetric fading statistics ($\gamma_k=1, \forall k$), since the average multicast capacity vanishes as $1/K$ for a large $K$\cite{jindal2006capacity}, the average sum content delivery rate converges to a constant, yielding a non-scalable system. This negative result calls for a careful design of content delivery that benefits from the time varying nature of the underlying fading broadcast channel. 

\section{Problem Formulation}\label{section:formulation}
In this section, we study the fading Gaussian broadcast channel where the transmitter wishes to convey $2^K-1$ mutually independent messages, each intended to a subset of users. We characterize the capacity region of these messages and then solve explicitly the weighted sum rate maximization problem. We show that this formulation allows to maximize the content delivery rate by opportunistically exploiting the wireless channel. 

\subsection{Broadcasting private and multiple common messages}
We start by observing that the channel at hand in \eqref{eq:FadingBC} for a given channel realization $\hv$ corresponds to a stochastically degraded Gaussian broadcast channel. Without loss of generality, let us assume $h_1\geq \dots \geq h_K$. % so that the following Markov chain holds.
% \[X \leftrightarrow Y_1  \leftrightarrow  \dots  \leftrightarrow  Y_K. \]
The capacity region of the degraded broadcast channel for $K$ private messages and a common message is well-known \cite{el2011network}.
 In this section, we consider a more general setup where the transmitter wishes to convey 
$2^K-1$ mutually independent messages, denoted by $\{M_{\Jc}\}$, where $M_{\Jc}$ denotes the message intended to the users in subset $\Jc\subseteq [K]$. Each user $k$ must decode all messages $\{M_{\Jc}\}$ for $\Jc\ni k$. By letting $R_{\Jc}$ denote the multicast rate of the message $M_{\Jc}$, we say that the rate-tuple $\Rm\in \RR_+^{2^K-1}$ is achievable if there exists encoding and decoding functions which guarantee a rate greater than $\Rm$. The capacity region is defined as the supremum of the achievable rate-tuple. Then we have the following result.   
\begin{theorem}\label{theorem:Region}
The capacity region $\Gamma(\hv)$ of a $K$-user degraded Gaussian broadcast channel with fading gains $h_1 \geq \dots \geq h_K$ and $2^K-1$ independent messages $\{M_{\Jc}\}$ is given by 
\begin{align}
R_1 & \leq \log(1+ h_1\alpha_1 P) \\ \label{eq:9}
\sum_{\Kc: k\in \Kc \subseteq [k]} R_{\Kc} & \leq \log\frac{1+ h_k \sum_{j=1}^{k} \alpha_j P}{1+ h_k\sum_{j=1}^{k-1} \alpha_j P} \;\;\; k=2, \dots, K \end{align}
for non-negative variables $\{\alpha_k\}$ such that $\sum_{k=1}^K \alpha_k \leq 1$. 
\end{theorem}
\begin{proof}
 \ifarxiv
	See Appendix \ref{appendix:theorem1}.
\else
	See \cite{techreport}.  
\fi
\end{proof}

The achievability builds on superposition coding at the transmitter and successive interference cancellation at receivers. For $K=3$, the transmit signal is simply given by 
\[x = x_1 + x_2 + x_3 + x_{12} + x_{23}+ x_{13}+ x_{123},
\]
where $\{x_{\Jc}\}$ are mutually independent and $x_{\Jc}\sim\Nc_{\CC}(0, \alpha_{\Jc} P)$ denotes the signal corresponding to the message $M_{\Jc}$ intended to the subset $\Jc\subseteq \{1,2,3\}$.   
User 3 (the weakest user) decodes $\tilde{M}_3 =\{M_{3}, M_{13}, M_{23}, M_{123}\}$ by treating all the other messages as noise. User 2 decodes first the messages $\tilde{M}_3$ and then jointly decodes $\tilde{M}_2 =\{M_2, M_{12}\}$. Finally, user 1 (the strongest user) decodes successively $\tilde{M}_3, \tilde{M}_2$ then finally $M_1$. 

\subsection{Weighted sum rate maximization}
In order to characterize the boundary of the capacity region $\Gamma(\hv)$, we consider the weighted sum rate maximization given as
\begin{align}\label{eq:generalWSR}
\max_{\rv \in \Gamma(\hv)}\sum_{\Jc: \Jc \subseteq [K]} \theta_{\Jc} r_{\Jc}.
\end{align}
By exploiting a simple property of the capacity region, the problem at hand can be cast into a simpler problem as summarized below. 
\begin{theorem}\label{theorem:WSR}
The weighted sum rate maximization with $2^K-1$ variables in \eqref{eq:generalWSR} reduces to a simpler problem with $K$ variables, given by 
\begin{align*}%\label{eq:powerallocation}
f(\alphav) = \sum_{k=1}^K \phi_k \log\frac{1+h_k \sum_{j=1}^{k} \alpha_j P}{1+ h_k\sum_{j=1}^{k-1} \alpha_j P}, 
\end{align*}
where $\phi_k$ denotes the largest weight for user $k$ 
\[
\phi_k\eqdef \max_{\Kc: k\in \Kc \subseteq [k]}\theta_{\Kc}.
\]
\end{theorem}
\begin{proof}
The proof builds on the simple structure of the capacity region. We first remark that for a given power allocation of other users, user $k$ sees $2^{k-1}$ messages $\{M_{\Jc}\}$ for $k\in \Jc \subseteq [k]$ with the equal channel gain. For a given power allocation $\alpha^{k}$, the capacity region of these messages is a simple hyperplane 
characterized by $2^{k-1}$ vertices $C_k \ev_i$ for $i=1, \dots, 2^{k-1}$, where $C_k$ is the sum rate of user $k$ in the RHS of \eqref{eq:9} and $\ev_i$ is a vector with one for the $i$-th entry and zero for the others. Therefore, the weighted sum rate is maximized for user $k$ by selecting the vertex corresponding to the largest weight, denoted by $\phi$. This holds for any $k$. 
\end{proof}

We provide an efficient algorithm to solve this power allocation problem as a special case of the parallel Gaussian broadcast channel studied in \cite[Theorem 3.2]{david1999optimal}. Following \cite{david1999optimal}, we define the rate utility function for user $k$ given by 
\begin{align*}
u_k(z)= \frac{\phi_k}{1/h_k+z}-\lambda,
\end{align*}
where $\lambda$ is a Lagrange multiplier. The optimal solution corresponds to selecting 
the user with the maximum rate utility at each $z$ and the resulting power allocation for user $k$ is given as 
\begin{align}\label{eq:optimalalpha}
 \alpha^*_k = \left\{ z:   [\max_j u_j(z) ]_+ = u_k(z)  \right \}/P,
\end{align}
with $\lambda$ satisfying $P=\left [ \max_k \frac{\phi_k}{\lambda} -\frac{1}{h_k} \right]_+$.

\subsection{Application example}

In this subsection, we consider the long-term average sum content delivery maximization as one of the applications of the weighted sum rate maximization solved previously. 
By treating a codeword intended to a subset $\Kc$ of users as a message intended to the same subset, i.e. $M_{\Kc}= V_{\Kc}$ in \eqref{eq:V} and assuming that these codewords for different subsets are all independent, the sum content delivery rate achieved by superposition coding can be written as the weighted sum rate: 
\begin{align*} %\label{eq:wsr_caching}
\sum_{\Kc: \Kc \subseteq [K]} \theta_{\Kc} R_{\Kc} \quad \text{with} \quad \theta_{\Kc} = \frac{ |\Kc| }{T(m,|\Kc|)},
\end{align*}
where $R_{\Kc}$ denotes the rate of message $M_{\Kc}$ satisfying the constraints in Theorem \ref{theorem:Region}. By noting that the weights depend only on the cardinality of $\Kc$ and that the function $k/T(m,k)$ is increasing in $k$, we have the following properties i) $\theta_{\Kc} = \theta_{\Kc'}, \;\; \forall \Kc, \Kc'  \text{~such that~} |\Kc|=| \Kc'|$, ii) $\theta_{\Kc} < \theta_{\Jc} , \;\; \forall \Kc \subset \Jc$. 

These properties readily imply that the effective weight of user $k$, denoted by $\phi_k$, is given by 
\begin{align*}%\label{eq:thetatilde}
\phi_k = \max_{\Jc: k\in \Jc \subseteq [k]} \theta_{\Jc} =\frac{k }{T(m,k)}.
\end{align*}
Following Theorem \ref{theorem:WSR}, the resulting sum delivery rate of superposition coding for a given channel state such that $h_1\geq \dots \geq h_K$ is given by
\begin{align*}
\Rsp(\hv) = \sum_{k=1}^K \frac{k}{T(m, k)} \log\left(1+\frac{h_{k} \alpha^*_k P}{1+h_{k} \sum_{j=1}^{k-1} \alpha^*_j P }\right),
\end{align*}
where $\{\alpha^*_j\}$ is the optimal power allocation in \eqref{eq:optimalalpha}. The long-term average sum delivery rate is given by 
\begin{align*}
\Rsp = \EE_{\hv} [ \Rsp(\hv) ].
\end{align*}
\section{Performance Analysis}\label{sec:perfom_analysis}
In this section, we analyze the long-term average sum delivery rate of the proposed scheme in two cases of interest: a) a user selection scheme that selects the best subset of users as a function of the channel state and the weights, b) naive coded caching (or baseline scheme) that applies coded caching to $K$ users as described in Section \ref{section:model}. By restricting ourselves to the symmetric fading case ($\gamma_k=1, \forall k$), we consider
two regimes of interest, i.e. large $K$ and high SNR. 

\subsection{Baseline scheme: \textit{naive} coded caching}
In this scheme, the server serves all $K$ users with the multicast rate limited by the worst user as in \eqref{eq:multicast_rate}. We define the exponential integral function $E_1(x)=\int_{1}^{+\infty} {e^{-xt}  \over t} dt$. The performance of this scheme is summarized below. 
\begin{proposition}\label{BaselineScheme}
	(i) $\Rbl (K,P) = \phi_K e^{{K \over P}} E_1\left({K \over P}\right)$.
	
	(ii) For all $P$: $\Rbl(K,P) \sim {P m \over  1-m}$ when $K \to \infty$.
	
	(iii) For all $K$: $\Rbl(K,P) \sim \phi_K \log(P)$ when $P \to \infty$.
\end{proposition}
\ifarxiv
\begin{proof}
See Appendix \ref{appendix:prop1}.
\end{proof}
\else
{\bf Sketch of proof: } The content delivery rate is:
\begin{align*}
\Rbl (K,P)= \phi_K \EE\left[ \log\left(1+ Ph_{\rm min} \right)\right],
\end{align*}
where $h_{\rm min}\defeq\min_{k=1, \dots, K} h_k$. Since $(h_k)_{k=1,...,K}$ are i.i.d. with distribution Exp($1$), $h_{\rm min}$ has distribution Exp($K$), so:
\begin{align*}
 	\EE\left[ \log\left(1+ Ph_{\rm min} \right)\right] &= \int_{0}^{+\infty} e^{-x} \log\left(1 + {P \over K}x\right) dx \\
 	&= e^{{K \over P}} E_1\left({K \over P}\right),
\end{align*}
which yields (i). In the regime $K\rightarrow \infty$, we have $\phi_K \sim \frac{Km}{1-m}$ as well as $\log(1+Px/K) \sim Px/K$, yielding (ii).  In the regime of high SNR $P\rightarrow \infty$, we use $E_1(x)\sim \log(1/x)$ for $x\rightarrow 0$, yielding (iii).  
\ep
\fi

\subsection{User selection scheme: opportunistic scheduling}

Albeit suboptimal, we consider a simple time-sharing strategy, which allocates a fraction of time $\eta_{\Kc}$ to the subset of users $\Kc$, with $\sum_{\Kc\subseteq [K]} \eta_{\Kc}=1$. 
The corresponding weighted sum rate maximization is given by 
\begin{align*}
\max_{\etav: \sum_{\Kc} \eta_{\Kc} =1} \sum_{\Kc\subseteq [K]} \theta_{\Kc} \eta_{\Kc} \log(1+P\min_{k\in\Kc}h_k). 
\end{align*}
 Let $\piv = \{\pi_1,\dots,\pi_K\}$ denote the permutation such that $h_{\pi_1} \ge \dots \ge h_{\pi_K}$.
Because of the capacity region structure, the problem at hand can be simplified into: 
\begin{align*}
\max_{\etav} \sum_{k=1}^K \phi_{k} \eta_{k} \log(1+h_{\pi_k}P). 
\end{align*}
The optimal solution is readily given by
\begin{align*}
\eta_k= \begin{cases}
1, & \text{if $k=  \arg\max_j \phi_{j}\log(1+h_{\pi_j}P)$},\\
0, & \text{otherwise}.
\end{cases}
\end{align*}
This means that we transmit to only one set of users maximizing the instantaneous weighted rate with full power. By transmitting opportunistically to the group of users with the highest sum content delivery rate at each channel realization, the long-term average sum content delivery rate is given by 
\begin{align*}
\Rsc(K,P) &= \EE\left[  \max_k \phi_k \log(1+h_{\pi_k} P) \right].
\end{align*}
We characterize $\Rsc(K, P)$ in two regimes of interest. 
\begin{proposition}\label{SelectionScheme}
(i) For all $P$: \newline $\Rsc(K,P) \sim {Km \over 1-m} e^{(\frac{1}{P}-{1 \over W(P)})} W(P)$ when $K \to \infty$, where $W(x)$ is the Lambert function i.e.  $W(x)e^{W(x)}=x$.

(ii) For all $K$: $\Rsc(K,P) \sim \phi_K \log(P)$ when $P \to \infty$.	
\end{proposition}

\ifarxiv
\begin{proof}
See Appendix \ref{appendix:prop2}.
\end{proof}
\fi

\begin{figure*}
  \begin{minipage}{\columnwidth}
\vspace{-8pt}
\begin{center}
\includegraphics[width=0.8\textwidth,clip=]{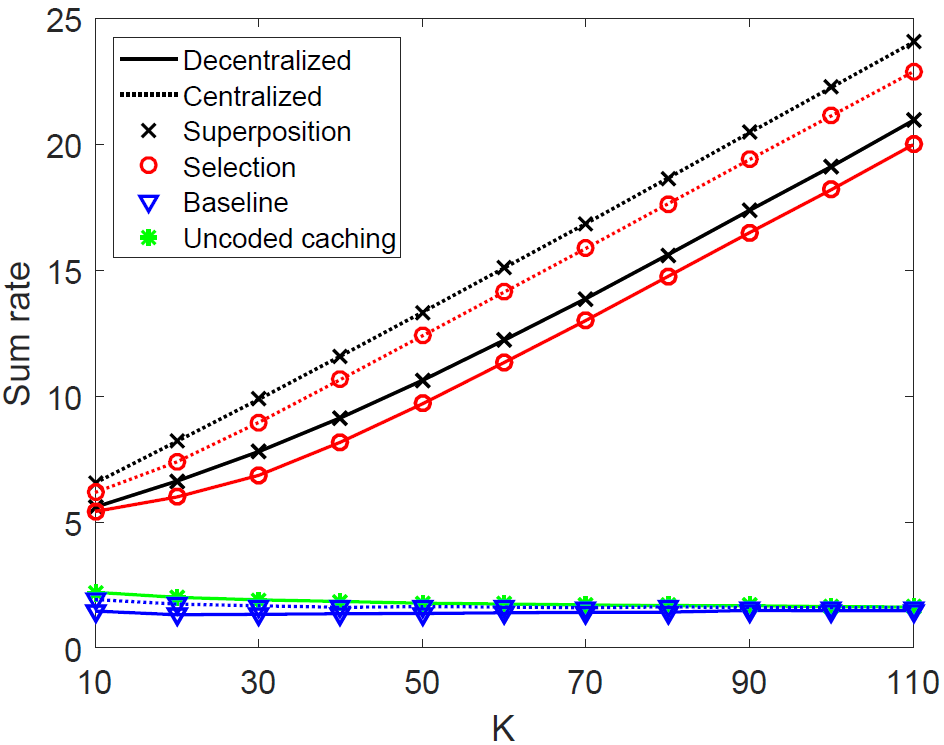}
\vspace{-5pt}
\caption{Sum rate vs $K$ for $P=10$(dB).}
\label{fig:users}
\end{center}
\vspace{-10pt}
\end{minipage}
%\vspace{-5pt}
  \begin{minipage}{\columnwidth}
  \vspace{-8pt}
\begin{center}
\includegraphics[width=0.8\textwidth,clip=]{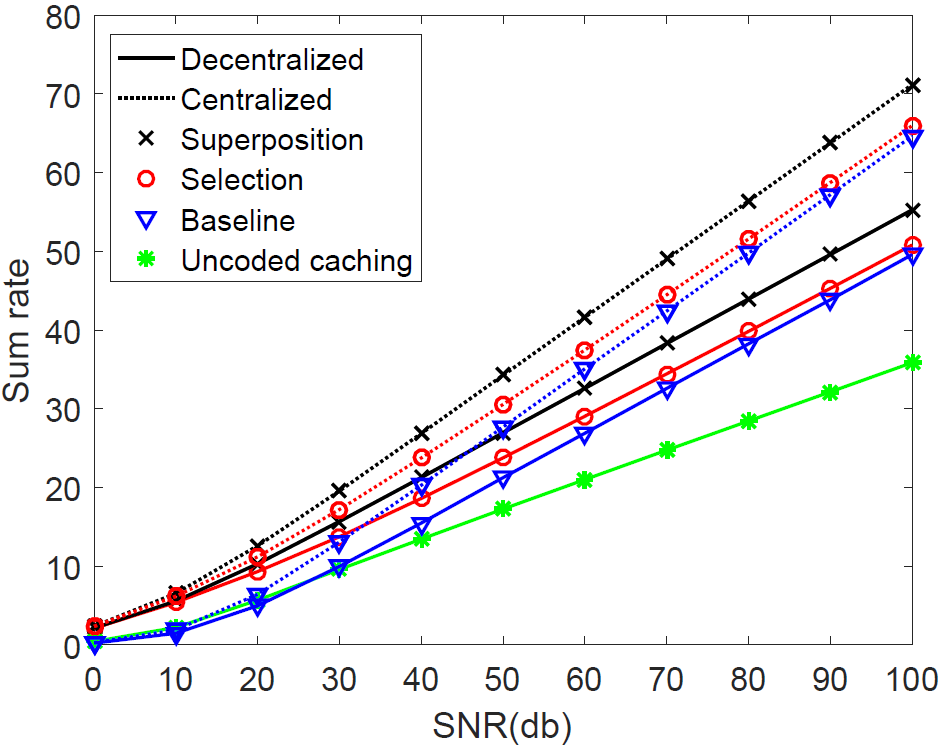}
\vspace{-5pt}
\caption{Sum rate vs SNR for $K=10$.}
\label{fig:SNR}
\end{center}
%\end{figure}
\vspace{-10pt}
\end{minipage}
%\vspace{-5pt}
\end{figure*} 

\subsection{Interpretation of the results}
%On one hand, it is proved that selection scheme achieves a linear increase of the average sum content delivery rate in the regime of a large $K$ thus yields a scalable solution. 
%On the other hand, both baseline scheme and selection scheme achieve the same sum delivery rate in the high SNR regime, since it is nearly optimal to perform coded caching over all $K$ users in this regime.  Moreover, we provide a simple threshold-based feedback scheme that mimics selection scheme, while requiring each user to feedback only one bit  rather than channel state information.  

From propositions \ref{BaselineScheme} and \ref{SelectionScheme}, the following remarks are in order: 1) in the large $K$ regime, the long-term sum delivery rate of the selection scheme grows linearly for any finite SNR. This is in a sharp contrast with the baseline scheme, whose sum delivery rate converges to a constant; 2) in the high SNR regime, both schemes yield the same performance, i.e. $\frac{K}{T(m, K)}\log P$, for any finite $K$ because the sum delivery rate is no longer sensitive to the randomness of channels and is maximized solely by exploiting the global caching gain; 3) It is worth noticing that the performance of selection scheme can be achieved without instantaneous channel knowledge. Namely, each user can measure its SNR and send a one-bit feedback indicating whether it is above or below the threshold value given by $ Pz^*=\frac{P}{W(P)}-1$.

\section{Numerical Examples}\label{sec:numericalEx}
In this section, we compare our proposed superposition scheme, its two special cases (baseline and selection), as well as uncoded caching. 
Uncoded caching refers to the case where the server sends the remaining $(1-m)F$ bits of the requested file at rate $\log(1+Ph_k)$ for each user $k$. Thus, the corresponding long-term average sum delivery rate is given by $$\EE\left[K \left(\sum_{k=1}^{K}\frac{1-m}{\log(1+Ph_k)}\right)^{-1} \right].$$ We consider a database of size $N=10^4$, normalized memory size of $m=10^{-1}$. %Under the two placement strategies (decentralized and centralized), we plot the sum content delivery rate provided by each scheme as a function of the number of users and the power constraint as depicted  and Fig. \ref{fig:SNR}, respectively.\\
In Fig. \ref{fig:users}, we plot the long-term sum content delivery rate as a function of the number of users at $P=10$ dB for both centralized (dashed line) and decentralized (solid line) placement strategies. We observe that both the superposition schemes and the selection scheme offer a linear increase, whereas the performance of baseline and uncoded schemes is bounded. This behavior agrees with the analysis of the previous section and implies that the performance of coded caching at low to moderate SNR is limited by the vanishing multicast rate. Furthermore, the selection scheme offers performance almost as good as the superposition scheme, despite its reduced complexity.

%We observe that uncoded caching scheme performs better than the baseline scheme, which implies that by using baseline scheme, we can %not benefit from coded caching gain in fading channel with low SNR. 
%We also observe that the gap between the two placement strategies for baseline scheme vanishes with the number of users. This is because decentralized and centralized placements are approximately the same for large $K$  as seen from \eqref{eq:defT}. 

In Fig. \ref{fig:SNR}, the long-term average sum content delivery rate is plotted as a function of SNR for different schemes. 
% and it shows again the good performance of our proposed scheme compared to the others.
  We observe that the performance of selection, baseline scheme becomes identical for large SNR, which confirms our analysis. In addition, the sum content delivery rate increases as SNR with a pre-log of $\phi_K$, which in turn depends on the placement strategy \eqref{eq:defT}. By comparing uncoded caching and the baseline scheme, we observe that after a certain SNR threshold, the baseline scheme performs better than uncoded caching scheme.
  
% There is a trade off between coded caching gain that scales with the number of users when considering perfect link, and multicast rate in fading channel that vanishes in $K$. In stead of applying coded caching to all users where the multicast capacity is limited by the weakest user, we proposed an efficient scheme that allows to benefit from coded caching under fading channel and outperforms selection scheme.

\section{Conclusion} 
We have studied content delivery using coded caching over fading broadcast channels. Contrary to the baseline scheme applying coded caching to $K$ users irrespectively of channel state information, we proposed opportunistic delivery schemes that achieve a linear increase of the sum content delivery rate by a careful selection of the user subset as a function of both channel state information and priorities. In order to reduce the amount and accuracy of feedback, we proposed a simple threshold-based feedback scheme yielding the same scalable solution while requiring only one bit per user. In future work, we plan on providing a detailed analysis of the performance of the more general superposition scheme proposed here.

\ifarxiv
\section{Appendix}
\subsection{Proof of Theorem \ref{theorem:Region}} \label{appendix:theorem1}
We provide the proof for $K=3$ and the general case $K>3$ follows readily. Here $\log$ denotes the binary logarithm,  $h(.)$ denotes the differential entropy, and $H(.)$ denotes Shannon entropy.  

\underline{Converse}
Notice that the channel output of user $k$ in \eqref{eq:FadingBC} for $n$ channel use can be equivalently written as 
\begin{align}
 y_{k,i}=x_{i}+\tilde{w}_{k,i}, ~~~~i=1,\dots,n
\end{align}
where $\tilde{w}_{k,i}=\frac{w_k[i]}{\sqrt{h_k}}\sim\Nc_{\CC}(0, N_k)$  for $N_k=\frac{1}{h_k}$. Since $N_1\leq N_2\leq N_3$, we set $\tilde{M}_k=\cup_{k\in\Jc\subseteq[k]} M_{\Jc}$ the message set that must be decoded by user $k$ at sum rate $\tilde{R}_k=\cup_{k\in\Jc\subseteq[k]} R_{\Jc}$. More explicitly, we have $\tilde{M}_1=\{M_1\}$, $\tilde{M}_2=\{M_2,M_{12}\}$, $\tilde{M}_3=\{M_3,M_{13},M_{23},M_{123}\}$.
By Fano's inequality, we have 
\begin{align}\label{fano}
\begin{cases}
nH(\tilde{M}_1)&\leq I(\tilde{M}_1;Y_1\cond \tilde{M}_2,\tilde{M}_3)\\
nH(\tilde{M}_2)&\leq I(\tilde{M}_2;Y_2\cond \tilde{M}_3)\\
nH(\tilde{M}_3)&\leq I(\tilde{M}_3;Y_3).
\end{cases}
\end{align}
Consider first user $3$.
\begin{align}\label{eq171}
 I(\tilde{M}_3;Y_3)=h(Y_3)-h(Y_3\cond \tilde{M}_3).
\end{align}
 Since we have $n\log\left(\pi eN_3\right)=h(Y_3\cond \tilde{M}_3,X) \leq h(Y_3\cond \tilde{M}_3)\leq h(Y_3)\leq n\log\left( \pi e(P+N_3)\right)$, there exists $0\leq \alpha_3\leq 1$ such that 
 \begin{align}\label{eq172}
 h(Y_3\cond \tilde{M}_3)=n\log\left(\pi e((1-\alpha_3)P+N_3)\right).
\end{align}
Using \eqref{eq171} and \eqref{eq172} we obtain
\begin{align}\label{eq173}
\MoveEqLeft{ I(\tilde{M}_3;Y_3)}\nonumber\\
&= h(Y_3)-h(Y_3\cond \tilde{M}_3)\nonumber\\
 &\leq n\log\left(\pi e(P+N_3)\right)-n\log\left(\pi e((1-\alpha_3)P+N_3)\right)\nonumber\\
 &=n\log\left(\frac{N_3+P}{N_3+(1-\alpha_3)P}\right).
\end{align}

 Next consider user $2$.
\begin{align}\label{eq174}
 I(\tilde{M}_2;Y_2\cond \tilde{M}_3)=h(Y_2\cond \tilde{M}_3)-h(Y_2\cond \tilde{M}_2,\tilde{M}_3).
\end{align}
Using the conditional entropy power inequality in \cite{el2011network} , we have
\begin{align}\label{eq175}
h(Y_3\cond \tilde{M}_3)&=h(Y_2+\tilde{W}_3-\tilde{W}_2\cond \tilde{M}_3)\nonumber\\
&\geq n\log(2^{h(Y_2\cond \tilde{M}_3)/n}+2^{h(\tilde{W}_3-\tilde{W}_2\cond \tilde{M}_3)/n})\nonumber\\
&= n\log(2^{h(Y_2\cond \tilde{M}_3)/n}+\pi e(N_3-N_2)).
\end{align}
\eqref{eq172} and \eqref{eq175} imply 
\begin{align}
\MoveEqLeft{n\log\left(\pi e((1-\alpha_3)P+N_3)\right)}\nonumber\\
&\geq n\log(2^{h(Y_2\cond \tilde{M}_3)/n}+\pi e(N_3-N_2))\nonumber
\end{align}
equivalent to
\begin{align}\label{eq176}
 h(Y_2\cond \tilde{M}_3)&\leq n\log(\pi e((1-\alpha_3)P+N_2)).
\end{align}
Since $n\log(\pi eN_2)=h(Y_2\cond \tilde{M}_2,\tilde{M}_3,X)\leq h(Y_2\cond \tilde{M}_2,\tilde{M}_3)\leq h(Y_2\cond \tilde{M}_3)\leq n\log(\pi e((1-\alpha_3)P+N_2))$, there exists $\alpha_2$ such that $0\leq1-\alpha_2-\alpha_3\leq 1-\alpha_3$ and 
\begin{align}\label{eq177}
h(Y_2\cond \tilde{M}_2,\tilde{M}_3)=n\log(\pi e((1-\alpha_2-\alpha_3)P+N_2)).
\end{align}
Using \eqref{eq174}, \eqref{eq176} and \eqref{eq177} it follows
\begin{align}\label{eq178}
 I(\tilde{M}_2;Y_2\cond \tilde{M}_3)&=h(Y_2\cond \tilde{M}_3)-h(Y_2\cond \tilde{M}_2,\tilde{M}_3)\nonumber\\
 &\leq n\log(\pi e((1-\alpha_3)P+N_2))\nonumber\\
 \MoveEqLeft{-n\log(\pi e((1-\alpha_2-\alpha_3)P+N_2))}\nonumber\\
 &=n\log\left(\frac{N_2+(1-\alpha_3)P}{N_2+(1-\alpha_2-\alpha_3)P}\right).
\end{align}

Finally we consider user $1$.
\begin{align}
 \MoveEqLeft{I(\tilde{M}_1;Y_1\cond \tilde{M}_2,\tilde{M}_3)}\nonumber\\
 &=h(Y_1\cond \tilde{M}_2,\tilde{M}_3)-h(Y_1\cond \tilde{M}_1, \tilde{M}_2,\tilde{M}_3)\nonumber\\
 &\leq h(Y_1\cond \tilde{M}_2,\tilde{M}_3)-h(Y_1\cond \tilde{M}_1, \tilde{M}_2,\tilde{M}_3, X)\nonumber\\
 &=h(Y_1\cond \tilde{M}_2,\tilde{M}_3)-h(Y_1\cond  X)\label{eq:markov}\\
 &=h(Y_1\cond \tilde{M}_2,\tilde{M}_3)-n\log\left(\pi eN_1\right)\label{eq179}
\end{align}
where \eqref{eq:markov} holds because $\left( \tilde{M}_1 ,\tilde{M}_2, \tilde{M}_3\right) \rightarrow X\rightarrow Y_1$ is a Markov chain.
Using the conditional entropy power inequality in \cite{el2011network} , we have
\begin{align}\label{eq180}
\MoveEqLeft{h(Y_2\cond \tilde{M}_2, \tilde{M}_3)}\nonumber\\
&=h(Y_1+\tilde{W}_2-\tilde{W}_1\cond \tilde{M}_2,\tilde{M}_3)\nonumber\\
&\geq n\log(2^{h(Y_1\cond \tilde{M}_2,\tilde{M}_3)/n}+2^{h(\tilde{W}_2-\tilde{W}_1\cond \tilde{M}_2,\tilde{M}_3)/n})\nonumber\\
&= n\log(2^{h(Y_1\cond \tilde{M}_2,\tilde{M}_3)/n}+\pi e(N_2-N_1))
\end{align}

\eqref{eq177} and \eqref{eq180} imply 
\begin{align}
\MoveEqLeft{n\log(\pi e((1-\alpha_2-\alpha_3)P+N_2))}\nonumber\\
&\geq n\log(2^{h(Y_1\cond \tilde{M}_2,\tilde{M}_3)/n}+\pi e(N_2-N_1))\nonumber
\end{align}
equivalent to
\begin{align}\label{eq181}
 h(Y_1\cond \tilde{M}_2,\tilde{M}_3)&\leq n\log(\pi e((1-\alpha_2-\alpha_3)P+N_1)).
\end{align}
Let $\alpha_1=1-\alpha_2-\alpha_3$. Combining the last inequality with \eqref{eq179} we obtain 
\begin{align}\label{eq182}
 I(\tilde{M}_1;Y_1\cond \tilde{M}_2,\tilde{M}_3)\leq n\log\left( \frac{N_1+\alpha_1P}{N_1}\right). 
\end{align}
From \eqref{fano}, \eqref{eq173}, \eqref{eq178} and \eqref{eq182}, it readily follows that $\exists$ $0\leq\alpha_1,\alpha_2,\alpha_3\leq1$ such that $\alpha_1+\alpha_2+\alpha_3=1$ and 

\begin{align}
\begin{cases}
H(\tilde{M}_1)&\leq \log\left(1+ \frac{\alpha_1P}{N_1}\right),\nonumber\\
H(\tilde{M}_2)&\leq \log\left(1+\frac{\alpha_2P}{N_2+\alpha_1P}\right),\nonumber\\
H(\tilde{M}_3)&\leq \log\left(1+\frac{\alpha_3P}{N_3+(\alpha_1+\alpha_2)P}\right).
\end{cases}
\end{align}
By replacing $H(\tilde{M}_k)$ with $\sum_{k\in\Kc\subseteq[k]}R_{\Kc}$ and $N_k$ with $\frac{1}{h_k}$ we obtain the result
\begin{align}
\begin{cases}
 R_1&\leq \log\left( 1+h_1\alpha_1P\right)\nonumber\\
R_2+R_{12}&\leq \log\left(\frac{1+h_2(\alpha_1+\alpha_2)P}{1+h_2\alpha_1P}\right)\nonumber\\
R_3+R_{13}+R_{23}+R_{123}&\leq \log\left(\frac{1+h_3P}{1+h_3(\alpha_1+\alpha_2)P}\right),
\end{cases}
\end{align}
%where $\tilde{R}_k=H(\tilde{M}_k)$.
\underline{Achievability} 
We prove that superposition coding achieves the upper bound. For $\Jc \subseteq \{1,2,3\}$, generate random sequences $x_{\Jc}(m_{\Jc})$, $m_{\Jc}\in[1:2^{nR_{\Jc}}]$ each i.i.d. $\Nc_{\CC}(0, \alpha_{\Jc}P)$, where $\sum_{\Jc\subseteq\{1,2,3\}}\alpha_{\Jc}=1$. We define $\tilde{x}_{k}(\tilde{m}_k)=\sum_{k\in\Jc\subseteq[k]} x_{\Jc}(m_{\Jc})$, where $\tilde{m}_{k}\in[1:2^{n\tilde{R}_{k}}]$. To transmit $\{m_{\Jc}\}_{\Jc \subseteq \{1,2,3\}}$, the encoder set $X=\sum_{\Jc\subseteq\{1,2,3\}}x_{\Jc}(m_{\Jc})=\tilde{x}_{1}(\tilde{m}_1)+\tilde{x}_{2}(\tilde{m}_2)+\tilde{x}_{3}(\tilde{m}_3)$.
 For decoding:
 \begin{itemize}
 \item Receiver $3$ jointly decodes $\{m_3,m_{13},m_{23},m_{123}\}$ by treating $\tilde{x}_{1}(\tilde{m}_1)$ and $\tilde{x}_{2}(\tilde{m}_2)$ as noise.
 \item Receiver $2$ uses successive cancellation by first decoding $\tilde{x}_{3}(\tilde{m}_3)$ and treating $\tilde{x}_{1}(\tilde{m}_1)$ and $\tilde{x}_{2}(\tilde{m}_2)$ as noise. It recovers $\{m_{23},m_{123}\}$. By subtracting off $\tilde{x}_{3}(\tilde{m}_3)$ and treating $\tilde{x}_{1}(\tilde{m}_1)$ as noise, user $2$ decodes $\tilde{x}_{2}(\tilde{m}_2)$ from which it recovers $\{m_2,m_{12}\}$.
  \item Receiver $1$ decodes $\tilde{x}_{3}(\tilde{m}_3)$ and recovers $\{m_{13},m_{123}\}$. Then, by successive cancellation it decodes $\tilde{x}_{2}(\tilde{m}_2)$ and recovers $\{m_{12}\}$. Finally it decodes $\tilde{x}_{1}(\tilde{m}_1)$ to recover $\{m_1\}$.
 \end{itemize}
\subsection{Proof of Proposition \ref{BaselineScheme}} \label{appendix:prop1}
The content delivery rate is:
\begin{align*}
\Rbl (K,P)= \phi_K \EE\left[ \log\left(1+ Ph_{\rm min} \right)\right],
\end{align*}
where $h_{\rm min}\defeq\min_{k=1, \dots, K} h_k$. Since $(h_k)_{k=1,...,K}$ are i.i.d. with distribution Exp($1$), $h_{\rm min}$ has distribution Exp($K$). Hence:
\begin{align*}
 	\EE\left[ \log\left(1+ Ph_{\rm min} \right)\right] &= \int_{0}^{+\infty} e^{-x} \log\left(1 + {P \over K}x\right) dx \\
 	&= e^{{K \over P}} E_1\left({K \over P}\right),
\end{align*}
which yields statement (i). 

When $K \to \infty$ we have $\phi_K \sim {K m \over 1-m}$ and
$$
\int_{0}^{+\infty} e^{-x} \log\left(1 + {P \over K}x\right) dx \sim {P \over K} \int_{0}^{+\infty} x e^{-x} dx = {P \over K},
$$
Replacing yields statement (ii): 
$$\Rbl(K,P) \sim {P m \over 1-m}.$$

When $P \to \infty$, ${K \over P} \to 0$. Since $E_1(x) \sim \log(1/x)$ for $x \to 0$ we obtain statement (iii):
$$\Rbl(K,P) \sim {\phi_K \log(P/K)} \sim  {\phi_K \log(P)}.$$

\subsection{Proof of proposition \ref{SelectionScheme}} \label{appendix:prop2}

We start by statement (i). The proof involves upper and lower bounding $\Rsc(K,P)$ by two expressions which are equivalent in the large $K$ regime. We define the complementary c.d.f. of $(h_k)_{k=1,...,K}$:
\begin{align*}
U(z) \defeq \sum_{k=1}^K \indic\{ h_k \ge z \},
\end{align*}
with $z \ge 0$. We further define the function: 
\begin{align*}
  g(z) \defeq {m \over 1-m} e^{-z} \log(1 + P z), z \ge 0.
\end{align*}
 It is noted that $g(0) = g(\infty) = 0$, and that $g$ is smooth. Differentiating, we have that $g$ is maximized at: 
 $$
 z^*(P) \defeq \arg\max_{z\geq0}g(z)={1 \over W(P)} - {1 \over P}
 $$
 so that:  
 $$
 \max_{z\geq0}g(z)=g(z^*)={m \over 1-m} e^{(\frac{1}{P}-{1 \over W(P)})} W(P).
 $$ 
The proof relies on the following equality:
\begin{align*}
\max_{k=1,...,K} \phi_k \log(1+h_{\pi_k} P) &= \max_{z \in \{h_1,...,h_K\}} \phi_{U(z)} \log(1+ P z) \\
	&= \max_{z \ge 0} \phi_{U(z)} \log(1+ P z).
\end{align*}
Indeed, function $z \mapsto \phi_{U(z)} \log(1+ P z)$ is left-continuous and is both continuous and increasing for all $z \not\in \{h_1,...,h_K\}$ so that it must attain its maximum in the set $\{h_1,...,h_K\}$.

\underline{Lower bound} Using the previous equality we obtain:
	\begin{align}
	\Rsc(K,P) &=\EE\left[ \max_{z\geq0}\phi_{U(z)} \log(1+ P z)\right] \nonumber\\
	&\geq \max_{z\geq0}\EE\left[\phi_{U(z)}\right] \log(1+ P z)\label{eq:jen1}\\
	&\geq \max_{z\geq0}\frac{m}{1-m}\EE\left[U(z)\right]\log(1+ P z) \label{eq:arg2}\\
	&= \max_{z\geq0}\frac{m}{1-m} K e^{-z}\log(1+ P z) \label{eq:arg3}\\
	&= K\max_{z\geq0}g(z)\nonumber\\
	&= {Km \over 1-m} e^{(\frac{1}{P}-{1 \over W(P)})} W(P),
\end{align}	  
where \eqref{eq:jen1} follows from Jensen's inequality; \eqref{eq:arg2} from the fact that $\phi_k \ge {m k \over 1-m}$ and \eqref{eq:arg3} from $\EE(U(z)) = K e^{-z}$.  

\underline{Upper bound} The upper bound is slightly more involved and involves a dominated convergence argument. Let us define: 
$$
G(K) = {1 \over K} \max_{z\geq0}\phi_{U(z)} \log(1+ P z),
$$
so that $\Rsc(K,P) = K \EE(G(K))$. We prove that:

(a) $\sup_{K} \EE(G(K)) <\infty$ and 

(b) $\limsup\limits_{K\rightarrow\infty} G(K) \stackrel{a.s.}{\le}  g(z^{*})$ 

If both (a) and (b) holds, applying the reverse Fatou lemma proves the announced result:
$$\limsup_{K\rightarrow\infty} {\Rsc(K,P) \over K} = \limsup_{K\rightarrow\infty} \EE(G(K)) \le g(z^*).$$

Consider claim (a). Since $\phi_k \le k$ $\forall k$:
\begin{align*}
	\phi_{U(z)} \log(1+ P z) &\le U(z) \log(1+ P z) \\
	&= \sum_{k=1}^K \indic\{ h_k \ge z\} \log(1+ P z) \\
	&\le \sum_{k=1}^K \log(1+ P h_k).
\end{align*}
The above holds for all $z$, and taking expectations:
\begin{align*}
\EE(G(K)) &= {1 \over K} \EE( \sup_{z \ge 0} \phi_{U(z)} \log(1+ P z)) \\
&\le \EE(\log(1+P h_k) < \infty.
\end{align*}
The above holds for all $K$, so that $\sup_{K} \EE(G(K)) <\infty$.

We turn to claim (b). Consider $y > z^*(P)$ fixed, whose value will be made precise afterwards.  Define intervals \newline $I_0 \defeq [0,y]$, $I_1 \defeq [y,\infty)$ and for $i\in\{0,1\}$, define:
\begin{align*}
G_i(K) = {1 \over K} \max_{z \in I_i} \{ \phi_{U(z)} \log(1+ P z) \},
\end{align*}
so that $G(K) = \max \left\lbrace G_0(K) , G_1(K)\right\rbrace $. To prove that $\limsup\limits_{K\rightarrow\infty} G(K) \le g(z^*)$ it is sufficient to prove that $\limsup\limits_{K\rightarrow\infty} G_i(K) \le g(z^*)$ for $i \in \{0,1\}$. 

Consider $G_0(K)$. For $z \in I_0$, we have $U(z) \ge U(y)$, so that:
$$
\phi_{U(z)} =  {U(z) \over T(m,U(z))} \le {U(z) \over T(m,U(y))}.
$$ 
Therefore:
\begin{align}
	G_0(K) \le {1 \over T(m,U(y))} \max_{z \in I_0} \left\{ {U(z) \over K} \log(1+ P z) \right\}.\label{eq:V0}
\end{align}
The Glivenko-Cantelli theorem states that:
$$
	\sup_{z \ge 0} \left|{U(z) \over K} - e^{-z} \right| \xrightarrow{a.s.} 0
$$
so that:
\begin{align}
	\max_{z \in I_0} & \left| {U(z) \over K} \log(1+ P z) - e^{-z} \log(1+ P z)\right| \nonumber \\
	\le & \max_{z \ge 0} \left|{U(z) \over K} - e^{-z}\right| \log(1 + P y) \xrightarrow{a.s.} 0\label{eq:lim1}
\end{align}
From the law of large numbers $U(y) \xrightarrow{a.s.} \infty$, so $T(m,U(y)) \xrightarrow{a.s.} {1-m \over m}$, together with \eqref{eq:V0} and \eqref{eq:lim1} it implies
\begin{align}
	\lim_{K\to\infty} \sup G_0(K) \stackrel{a.s.}{\le} \max_{0 \le z \le y} g(z). \label{limV0}
\end{align}
Now, consider $G_1(K)$. For $z \in I_1$, by the same argument as previously:
\begin{align*}
	{1 \over K} \phi_{U(z)} \log(1+ P z) &\le {1 \over K} \sum_{k=1}^K \indic\{h_k \ge z\}\log( 1 + P h_k) \\ 
	&\le {1 \over K} \sum_{k=1}^K \indic\{h_k \ge y\}\log( 1 + P h_k) \\
	&\stackrel{a.s.}{\to} \EE( \indic\{h_k \ge y\} \log( 1 + P h_k)),
\end{align*}
using the law of large numbers. \newline Since $y \to \EE( \indic\{h_k \ge y\} \log( 1 + P h_k))$ is decreasing and vanishes when $y \to \infty$, we may select $y$ large enough so that:
$$
	\EE( \indic\{h_k \ge y\} \log( 1 + P h_k)) \le g(z^*).
$$
Putting it together $\lim\sup G_1(K) \stackrel{a.s.}{\le} g(z^*)$ which is claim (b). This concludes the proof of statement (i).

\hspace{1cm}

Consider statment (ii), we have for $P \to \infty$:
	\begin{align*}
		{\max_k \phi_k \log(1+h_{\pi_k} P) \over \log(P)} \stackrel{a.s.}{\to} \phi_K.
	\end{align*}	
	Furthermore, 
	$$
			\sup_{P \ge 0} \EE\left({\max_k \phi_k \log(1+h_{\pi_k} P) \over \log(P)}\right) = \sup_{P \ge 0} {\Rsc(K,P) \over \log(P)} < \infty
	$$
	so by Lebesgue's theorem ${\Rsc(K,P) \over \log(P)} \to \phi_K$.

\fi

\end{document}